\documentclass[a4paper,12pt]{article}
\usepackage{amssymb,amsmath,amsthm}
\usepackage{texdraw}
\usepackage{a4wide}
\usepackage{pgf,tikz}
\usetikzlibrary{arrows}

\usepackage{caption}
\usepackage{subcaption}

\newtheorem{theorem}{Theorem}[section]

\newtheorem{proposition}[theorem]{Proposition}

\theoremstyle{definition}

\newtheorem{remark}[theorem]{Remark}

\begin{document}


\title{ Discrete  integrable systems associated with relativistic collisions}
\author{Theodoros E. Kouloukas 
\\
\\School of Mathematics and Physics, \\ University of Lincoln, UK}

\date{ }
\maketitle

\begin{abstract}
We study vector quadrirational Yang--Baxter maps representing the momentum-energy transformation of 
two particles after elastic relativistic collisions. The collision maps admit Lax representations compatible with an $r$-matrix Poisson structure and 
correspond to integrable systems of quadrilateral lattice equations.

\end{abstract}

\bigskip

\section{Introduction}
The integrability of discrete systems (systems of ordinary or 
partial difference equations) is closely related to the concept 
of multidimensional consistency. This can be depicted as the three-dimensional (3D) consistency for equations on two-dimensional lattices \cite{ABS1,bobsur,NW}, with fields assigned to the vertices of elementary quadrilaterals, and as the (set theoretical) Yang--Baxter equation \cite{Baxter,buch,Drin,skly88,ves2,Yang} for 
maps with fields assigned to the edges of the quadrilaterals. Both the 3D consistency and the Yang--Baxter property reflect the compatibility of the equations/maps when extended to the three-dimensional lattice, and they are related to typical integrability features such as zero-curvature representations, {B\"acklund--Darboux} transformation, invariant Poisson structures,  recursion operators, symmetries  and conserved quantities (see e.g. \cite{ABS2,AdYa,sokor,kp3,MWX,Ni,paptonves,ves2}). 

In this article we study particle collision problems in the framework of discrete integrable systems. In \cite{Koul}, 
it was shown that the velocity transformation of two head-on elastically colliding particles satisfies the Yang--Baxter equation with the two masses acting as the Yang--Baxter parameters.  In the classical (non-relativistic) case the collision Yang--Baxter map is linear and is linked to a discrete wave equation (see  the relevant discussion in section \ref{clascol}). In the relativistic case, the induced collision map is a non-rational parametric Yang--Baxter map which can be transformed, under a non-rational change of variables, to a rational one associated with well-known integrable lattice equations of KdV type. This surprising link between discrete integrable systems and particle collision problems motivates this work. 

The aim of this paper is to investigate various integrability aspects of relativistic collision problems which extend the results of \cite{Koul}. We will introduce and study integrable, in the sense of multidimensional consistency,  birational maps and affine linear equations associated with elastic collisions. In the non-relativistic case the velocity 
transformation under the collision of two particles generates  linear systems that satisfy all the desirable properties. However, in the relativistic case we show that the momentum-energy vectors, rather than velocities, form a more natural set of variables. The momentum-energy  transformation of the colliding particles is a vector quadrirational  Yang--Baxter map which corresponds to an affine linear system of 3D consistent quadrilateral equations. Moreover, a higher-dimensional generalisation  associated with planar collisions is presented, which admits a Lax representation compatible with the Sklyanin bracket. Thus we derive invariant Poisson structures for transfer maps, which represent particular sequences of periodic colliding particles, on the two-dimensional lattice and Poisson commuting first integrals.

Section \ref{SecNot} includes a short introduction to the theory of  
3D consistent equations and Yang--Baxter maps. 
As a first application, we present the linear systems associated with classical elastic collisions. In section \ref{SecRel1}, we study head-on relativistic collisions. We review the velocity transformation map and introduce the momentum-energy transformation Yang--Baxter map, its Lax representation and the corresponding 3D consistent lattice equations. Higher dimensional  generalisations of these systems are presented in section \ref{SecHigherGen}, where we also study invariant Poisson structures, reductions and 
the Liouville integrability of the transfer maps. We conclude in section \ref{SecConc} with further comments and perspectives for future work.

\section{Integrable lattice equations, Yang--Baxter maps and classical collisions} \label{SecNot}

We consider an equation of the form $Q(f,f_{1},f_2,f_{12};\alpha,\beta)=0$, where the four variables  $f,f_{1},f_2,f_{12}\in \mathbb{C}^n$ are assigned to the four vertices of a quadrilateral and the parameters $\alpha,\beta$ are assigned to its edges (Figure \ref{fig1}). We assume that this equation is affine linear, that is linear with respect to any one of the arguments $f,f_1,f_2$ and $f_{12}$. Next, we consider fixed values $f,f_1,f_2,f_3$ at four vertices of a cube as in Figure \ref{fig1} (black points). 
By employing the same equation at the corresponding vertices of the down, front and left faces of the cube we can determine uniquely 
the values $f_{12},f_{13}$ and $f_{23}$. Hence, we can determine the value $f_{123}$ in three different ways by the rest three faces of the cube.  If all these three values coincide, i.e. $f_{123}$ is uniquely defined as a function of $f,f_1,f_2,f_3$,
then the equation $Q(f,f_{1},f_2,f_{12};\alpha,\beta)=0$ is called {\it {3D consistent}} or {\it consistent around the cube}   {\cite{ABS1,bobsur,NW}}.

\begin{figure}
\begin{center}
\includegraphics[width=0.32\linewidth]{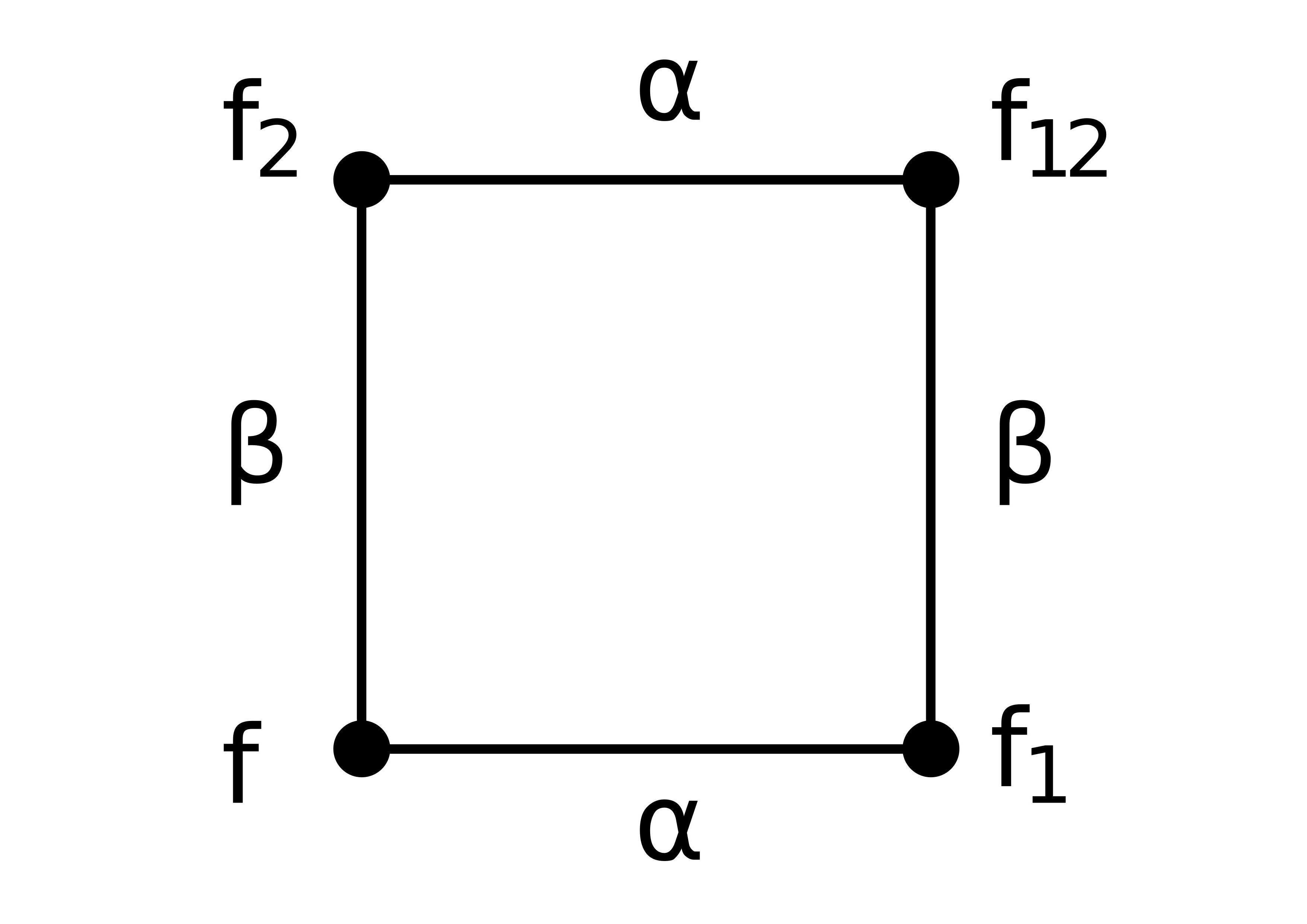} 
\includegraphics[width=0.42\linewidth]{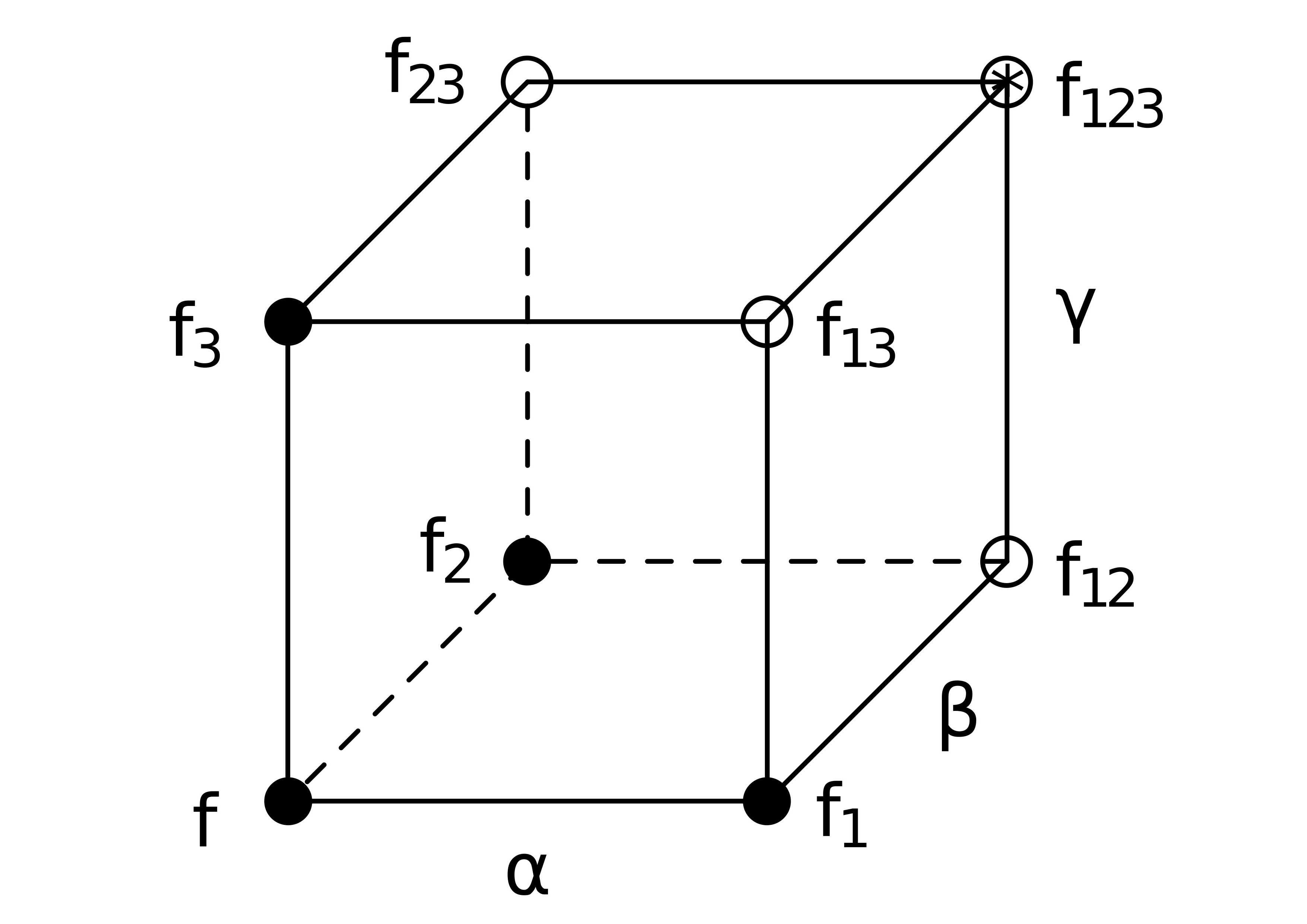} 
\caption{3D consistency} \label{fig1}
\end{center}
\end{figure}

\newpage

We can regard the equation $Q(f,f_{1},f_2,f_{12};\alpha,\beta)=0$ as defined on a two-dimensional quadrilateral lattice with fields $f_{i,j}:\mathbb{Z}^2 \rightarrow \mathbb{C}^n$, by setting $f=f_{n,m}, f_1=f_{n+1,m}, f_2=f_{n,m+1}$ and  $f_{12}=f_{n+1,m+1}$, thus
\begin{equation} \label{latEq} 
Q(f_{n,m},f_{n+1,m},f_{n,m+1},f_{n+1,m+1};\alpha,\beta)=0.
\end{equation}
The 3D consistency  then indicates that the lattice equation \eqref{latEq} can be embedded in a three dimensional lattice in a compatible way. Adler, Bobenko and Suris presented in \cite{ABS1} a classification of 3D consistent equations up to common M\"obius transformations of the variables.
An equivalent formulation of 3D consistency can be traced back in \cite{AdYa} 
(see also \cite{kp4,paptonves}). 

The equivalent of 3D consistency in the case of maps defined on the edges of a quadrilateral is the Yang--Baxter equation \cite{Baxter,  Drin, Yang}. Following   {\cite{buch,ves2}}, we will call a map $R: \mathcal{X} \times
\mathcal{X} \rightarrow \mathcal{X} \times \mathcal{X}$, with 
$R:(x,y)\mapsto (u(x,y),v(x,y))$, a {\it Yang--Baxter map} if it satisfies the {\it set-theoretical Yang--Baxter equation}, 
\begin{equation}\label{YBprop}
 R_{23}\circ R_{13}\circ
R_{12}=R_{12}\circ R_{13}\circ R_{23},
\end{equation} where $R_{ij}$,
for $i,j=1,2,3$, denotes the action of the map $R$ on the $i$ and
$j$ factor of $\mathcal{X} \times \mathcal{X} \times \mathcal{X}$,
i.e. $R_{12}(x,y,z)=(u(x,y),v(x,y),z)$,
$R_{13}(x,y,z)=(u(x,z),y,v(x,z))$ and
$R_{23}(x,y,z)=(x,u(y,z),v(y,z))$. In general 
$\mathcal{X}$ can be any set, however here we will  
regard $\mathcal{X}$ as an algebraic 
variety. Furthermore, we will call the map $R$ {\it quadrirational}, if both maps $u( \cdot ,y)$, $v(x, \cdot )$, for fixed $y$ and $x$ respectively, are birational isomorphisms  of $\mathcal{X}$ to itself \cite{ABS2}. 
By considering $R$ as a map on the edges of an elementary quadrilateral, we can interpret the Yang--Baxter equation \eqref{YBprop} as the compatibility of the map embedded on the faces of a 3D cube as in Figure \ref{fig2}.

\begin{figure}
\begin{center}
\includegraphics[width=0.62\linewidth]{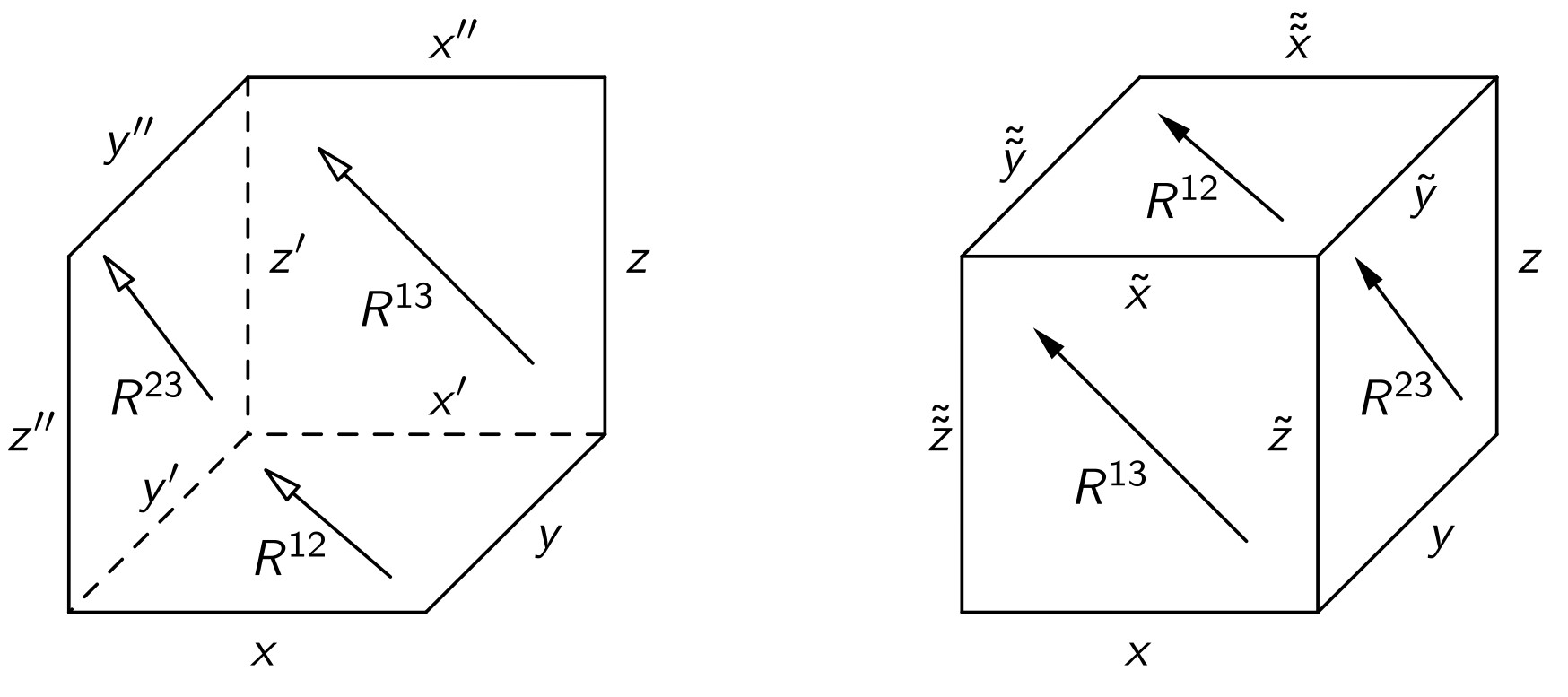} 
\caption{The cubic representation of the Yang--Baxter equation} \label{fig2}
\end{center}
\end{figure}

A \emph{parametric Yang--Baxter map} \cite{ves2,ves3} is a Yang--Baxter map  
 $R:(\mathcal{X} \times
\mathcal{I}) \times (\mathcal{X} \times \mathcal{I}) \mapsto
(\mathcal{X} \times \mathcal{I}) \times (\mathcal{X} \times
\mathcal{I})$, with
\begin{equation} \label{pYB}
R:((x,\alpha),(y,\beta))\mapsto((u,\alpha),(v,\beta))= ((u(x,\alpha,
y,\beta),\alpha),(v(x,\alpha, y,\beta),\beta)). 
\end{equation}
So, the Yang--Baxter parameters $\alpha, \beta \in \mathcal{I}$ are considered as extra variables that remain invariant 
under the map $R$. 
We usually
keep the parameters $\alpha, \beta \in \mathcal{I}$ separate and denote (\ref{pYB}) just by 
$R_{\alpha,\beta}:\mathcal{X}\times \mathcal{X} \rightarrow
\mathcal{X} \times \mathcal{X}$. 
Classifications of parametric Yang--Baxter maps on $\mathbb{CP}^1 \times \mathbb{CP}^1$ have been presented 
in \cite{ABS2,papclas}. 

A matrix $L$ that depends on a point $x \in \mathcal{X}$,  
a parameter $\alpha \in \mathcal{I}$ and a spectral parameter $\zeta\in \mathbb{C}$, such
that
\begin{equation} \label{laxmat}
L(u,\alpha,\zeta)L(v,\beta,\zeta)=L(y,\beta,\zeta)L(x,\alpha,\zeta)
\end{equation}
is called a {\em Lax matrix} of the Yang--Baxter map $R_{\alpha,\beta}$ \cite{ves4,ves2,ves3}. 
On the other hand, 
if $u=u_{\alpha,\beta}(x,y)$ and $v=v_{\alpha,\beta}(x,y)$ satisfy
(\ref{laxmat}) for a matrix $L$  and the equation $$L(
\hat{x},\alpha,\zeta )L( \hat{y} ,\beta,\zeta )L(\hat{z}, \gamma,\zeta )= L(x
,\alpha,\zeta)L(y, \beta,\zeta)L(z, \gamma,\zeta)$$ implies the unique solution $\hat{x}=x, \
\hat{y}=y$ and $\hat{z}=z$ for every $x,y,z \in \mathcal{X}$, then it follows that the map 
$R_{\alpha,\beta}:(x,y)\mapsto(u,v)$ is a Yang--Baxter map with Lax
matrix $L$ \cite{kp1,ves2}.   

3D consistent equations generate solutions of the Yang--Baxter equation and vice versa. This connection relies on the symmetries of the equations and the invariant conditions of the maps \cite{KasNi,kp4,paptong,paptonves,shib}.   { A similar approach can be considered for multi-component systems of difference equations \cite{KasNiPT}}. The following  proposition for Yang--Baxter maps defined on quasigroups which satisfy an invariant condition appears in  \cite{shib} (see also  \cite{kp4} for applications in the case of parametric Yang--Baxter maps). 
\begin{proposition} \label{YBto3D}
Let $R:L \times L \rightarrow L \times L $, $R(x,y)=(u(x,u),v(x,y))$ be a Yang--Baxter map on the quasigroup $(L,*)$ with $v(x,y)*u(x,y)=x*y$. Then 
\begin{equation} \label{Shib}
f_{n,m+1}=f_{n,m}*v(f_{n,m} \backslash f_{n+1,m},f_{n+1,m} \backslash f_{n+1,m+1})\;,
\end{equation}
where $\backslash$ denotes the left division operation, is a 3D consistent equation. 
\end{proposition}
 In the case of an abelian group $(L,+)$, the invariant condition becomes $u(x,y)+v(x,y)=x+y$ and the corresponding 3D consistent equation 
\begin{equation} \label{Shibf}
f_{n,m+1}=f_{n,m}+v(f_{n+1,m}-f_{n,m},f_{n+1,m+1}-f_{n+1,m})\;. 
\end{equation}

\subsection{Classical head-on collision systems} \label{clascol}
As a first example, we consider the map $R^0_{m_1,m_2}:(v_1,v_2)\mapsto(v_1',v_2')$, with  
\begin{equation} \label{clv}
v_1'=\frac{v_1(m_1-m_2)+2 m_2 v_2}{m_1+m_2}, \  v_2'=\frac{v_2(m_2-m_1)+2 m_1 v_1}{m_1+m_2}\;.
\end{equation}
This linear map represents the transformation of velocities of two particles with masses $m_1$, $m_2$ respectively after elastic (non-relativistic) collision. Here $v_1$, $v_2$ denote the initial velocities of the two particles and $v_1'$, $v_2'$  the corresponding velocities after the collision. 

The linear map $R^0_{m_1,m_2}$ satisfies the parametric Yang--Baxter equation \cite{Koul} and the invariant condition
$$m_1v_1+m_2v_2=m_1v_1'+m_2v_2'\;,$$ which reflects the conservation of momentum. 
The Yang--Baxter map $R^0_{m_1,m_2}$ appears in various different contexts in literature. In \cite{DimHos}, it describes the change of polarization throughout the tropical limit graph of soliton solutions of the vector KdV equation,  
while in \cite{KasNi} appears as a limit of Hirota's KdV equation. 

Now, if we consider vertex variables $w_{k,l}$ with  
\begin{align*}
&v_1=m_2(w_{k+1,l}-w_{k,l})\;, \ v_1=m_1(w_{k+1,l+1}-w_{k+1,l})\;, \\
& v_1'=m_2(w_{k+1,l+1}-w_{k,l+1})\;, \ v_2'=m_1(w_{k,l+1}-w_{k,l})\;, 
\end{align*}
then equations \eqref{clv} imply 
the discrete linear wave equation: 
\begin{equation} \label{dwave}
(m_2-m_1)(w_{n+1,m+1}-w_{n,m})+(m_2+m_1) (w_{n+1,m}-w_{n,m+1})=0\;.
\end{equation}
We can show directly that \eqref{dwave} is a 
3D consistent equation. Linear quadrilateral 3D consistent equations are described in \cite{atk} and linear Yang--Baxter maps in \cite{buchIg}.

\section{Head-on relativistic collisions and integrable lattice systems} \label{SecRel1}

In this section we present the Yang--Baxter maps and the  integrable lattice equations associated with head-on relativistic collisions. The parametric Yang--Baxter map which corresponds to the velocity transformation is not a rational map but it is equivalent, under a non-rational change of variables which preserves the Yang--Baxter property, to a quadrirational one associated with the discrete modified and the Schwarzian KdV equations. On the other hand, the 
map that represents the momentum-energy transformation after collision is by default quadrirational. This (non-parametric)  higher-dimensional map satisfies the Yang--Baxter equation too, and corresponds to a system of 3D consistent lattice equations.  

\subsection{Discrete systems associated with velocity transformations}
In \cite{Koul}, we have shown that the transformation of the velocities of two head-on  colliding particles with invariant masses $m_1$, $m_2$, after an elastic relativistic collision is given by the {\it collision map}
\begin{equation} \label{colYB}
\mathcal{R}_{m_1,m_2}= (\phi^{-1} \times \phi^{-1}) \circ {R}_{m_1,m_2} \circ (\phi \times \phi), 
\end{equation} 
where $\phi:(-c,c) \rightarrow (0,+ \infty)$ is the bijection $\phi(v)=\sqrt{\frac{c+v}{c-v}}$, 
$(\phi \times \phi)(v_1,v_2)=(\phi(v_1),\phi(v_2))$, $c$ is the speed of light and $R_{m_1,m_2}$ the quadrirational  map 
\begin{equation} \label{YBab}
{R}_{m_1,m_2}(x,y)=\left(y \left( \frac{m_1 x+m_2 y}{m_2 x+m_1 y}\right),x \left(\frac{m_1 x+m_2 y}{m_2 x+m_1 y} \right) \right):=(u,v)\;.   
\end{equation} 
Both maps, $R_{m_1,m_2}$ and $\mathcal{R}_{m_1,m_2}$ are parametric Yang--Baxter maps with Lax matrices 
\begin{equation} \label{Lax}
L(x,\alpha,\zeta)=\left(
\begin{array}{cc}
  \zeta & \alpha x \\
 \frac{\alpha}{x} &  \zeta
\end{array}
\right) \ \text{and} \ \mathcal{L}(x,m_1,\zeta)=L(\phi(x),m_1,\zeta)
\end{equation}
respectively. 

The collision map \eqref{colYB} is not rational but it is equivalent under the non-rational transformation $\phi$ to the quadrirational Yang--Baxter map \eqref{YBab} (for the Yang--Baxter equivalence see \cite{papclas}). The latter map corresponds to 
the $H_{III}^{A}$ map of the classification list in \cite{papclas}  
under the transformation $x \mapsto \alpha \hat{x}$, $y \mapsto \beta \hat{y}$,  $u \mapsto \alpha \hat{u}$, $v \mapsto \beta \hat{v}$
and the reparametrization $\alpha^2=\hat{\alpha}, \ \beta^2=\hat{\beta}$, 
and to  the $F_{III}$ map of the classification list in  
\cite{ABS2}, for 
$x \mapsto -\alpha x$, $y \mapsto \beta \hat{y}$, $u \mapsto \alpha \hat{u}$, $v \mapsto -\beta \hat{v}$ 
and the same reparametrization.

Now, from \eqref{YBab}, by setting  $x=w_{k+1,l}w_{k,l}$, $y=w_{k+1,l+1}w_{k+1,l}$, $u=w_{k+1,l+1}w_{k,l+1}$ and $v=w_{k,l+1}w_{k,l}$, we come up with the 3D consistent equation
\begin{equation} \label{mKdV}
w_{k+1,l}(\alpha w_{k,l}+\beta w_{k+1,l+1})-w_{k,l+1}(\beta w_{k,l}+\alpha w_{k+1,l+1})=0\;,
\end{equation}
which is a discrete version of the potential modified KdV equation \cite{NQC} and corresponds under 
a gauge transformation to $H_3$ equation in \cite{ABS1} (for $\delta=0$).
We can derive this equation directly from Proposition \ref{YBto3D} by considering the quasigroup $L=\mathbb{C}\backslash \{0\}$ with binary operation $x*y=\frac{y}{x}$, for any $x,y \in L$, (so that the left division is then defined as $x \backslash y=xy$), which also proves the 3D consistency of this equation. 

There is one more 3D consistent equation associated with the Yang--Baxter map \eqref{YBab} and the invariant condition $m_1u+m_2v=m_1x+m_2y$, that is the discrete Schwarzian KdV or cross-ratio lattice equation ($Q_1$ for $\delta=0$ of the ABS classification list \cite{ABS1}),  
\begin{equation} \label{SKdV}
\frac{(w_{k,l}-w_{k,l+1})(w_{k+1,l+1}-w_{k+1,l})}{(w_{k,l+1}-w_{k+1,l+1})(w_{k+1,l}-w_{k,l})}=\frac{m_2^2}{m_1^2}\;,
\end{equation}
which is obtained by considering vertex variables $w_{k,l}$, such that  $m_2(w_{k+1,l}-w_{k,l})=x$, $m_1(w_{k+1,l+1}-w_{k+1,l})=y$, $m_2(w_{k+1,l+1}-w_{k,l+1})=u$ and $m_1(w_{k,l+1}-w_{k,l})=v$.

\subsection{Discrete systems associated with  momentum-energy \\  transformations}

We consider the two colliding particles with rest masses $m_1$, $m_2$ and initial momentum-energy vectors $$\mathbf{x}=(x_0,x_1)=(\frac{E_x}{c},p_x)\;, \   \mathbf{y}=(y_0,y_1)=(\frac{E_y}{c},p_y)\;$$ respectively. Here, $E_x,E_y$ denote the relativistic energy of the two particles before collision and $p_x,p_y$ their momenta. We also denote by 
$$\mathbf{u}=(u_0,u_1)=(\frac{E_u}{c},p_u)\;, \   \mathbf{v}=(v_0,v_1)=(\frac{E_v}{c},p_v)\;$$
the momentum-energy vectors of the particles after collision, with  $E_u,E_v$ and $p_u, p_v$ their corresponding energies and momenta. The conservation of relativistic energy and momentum then reads 
\begin{equation} \label{invcond1}
\mathbf{u}+\mathbf{v}=\mathbf{x}+\mathbf{y}\;.
\end{equation}
Furthermore, the energy-momentum relation, $E^2=(pc^2)+(mc^2)^2$, for each particle implies that
\begin{equation} \label{invcond2}
u_0^2-u_1^2=x_0^2-x_1^2=m_1^2c^2\;, \  v_0^2-v_1^2=y_0^2-y_1^2=m_2^2c^2\;.
\end{equation}
The system of \eqref{invcond1} and \eqref{invcond2} admits two solutions with respect to $\mathbf{u}$ and $\mathbf{v}$, the trivial one $\mathbf{u}=\mathbf{x}$,  $\mathbf{v}=\mathbf{y}$, which corresponds to no-collision, and the collision solution 
\begin{align} \label{uv}
\mathbf{u}=\mathbf{y}+k(\mathbf{x},\mathbf{y})(\mathbf{x}+\mathbf{y}) \;,
\mathbf{v}=\mathbf{x}-k(\mathbf{x},\mathbf{y})(\mathbf{x}+\mathbf{y})\;, 
\end{align}
with  
\begin{equation} \label{K}
k(\mathbf{x},\mathbf{y})=\frac{({x}_0 ^2-{x}_1 ^2)-({y}_0^2-{y}_1^2)}{(x_0+y_0)^2-(x_1+y_1)^2} =\frac{\langle\mathbf{x},\mathbf{x}\rangle-\langle\mathbf{y},\mathbf{y}\rangle}{\langle\mathbf{x}+\mathbf{y},\mathbf{x}+\mathbf{y}\rangle}\;.
\end{equation}
Here, by $\langle \ , \ \rangle$ we denote the bilinear form  $\langle\mathbf{x},\mathbf{y}\rangle:=x_0 y_0-x_1 y_1$.  Hence, the invariant conditions   \eqref{invcond1}-\eqref{invcond2} imply that $\langle\mathbf{x},\mathbf{x}\rangle=\langle\mathbf{u},\mathbf{u}\rangle$, $\langle\mathbf{y},\mathbf{y}\rangle=\langle\mathbf{v},\mathbf{v}\rangle$ and $k(\mathbf{x},\mathbf{y})=k(\mathbf{u},\mathbf{v})$. 

Now, we define as the {\it momentum-energy map} the map 
\begin{equation} \label{MEmap}
\mathbf{R}:(\mathbf{x},\mathbf{y})\mapsto (\mathbf{u},\mathbf{v}):=(\mathbf{y}+k(\mathbf{x},\mathbf{y})(\mathbf{x}+\mathbf{y}),\mathbf{x}-k(\mathbf{x},\mathbf{y})(\mathbf{x}+\mathbf{y}))\;, 
\end{equation}
which maps the initial collision momentum-energy vectors of the two particles to the corresponding vectors after collision.

\begin{proposition} \label{thm1}
The momentum-energy map \eqref{MEmap} is a (non-parametric) quadrirational Yang--Baxter map with Lax matrix 
\begin{equation} \label{LaxME}
\mathbf{L}(\mathbf{x},\zeta)=\left(
\begin{array}{cc}
  \zeta & \ x_0+x_1 \\
 x_0-x_1 &  \zeta
\end{array}
\right).
\end{equation}

\end{proposition}

\begin{proof}
The Yang--Baxter property of the map $\mathbf{R}$ can be checked directly. We can also verify that for $\mathbf{u}=\mathbf{y}+K(\mathbf{x}+\mathbf{y})$ and $\mathbf{v}=\mathbf{x}-K(\mathbf{x}+\mathbf{y})$, 
$$\mathbf{L}(\mathbf{u},\zeta)\mathbf{L}(\mathbf{v},\zeta)=
\mathbf{L}(\mathbf{y},\zeta)\mathbf{L}(\mathbf{x},\zeta)\;,$$
which shows that $\mathbf{L}(\mathbf{x},\zeta)$ is a Lax matrix of  $\mathbf{R}$. Finally, the quadrirationality of $\mathbf{R}$ follows by observing that system  \eqref{uv} is equivalent to 
\begin{equation*}
\mathbf{u}=-\mathbf{v}+k(\mathbf{x},-\mathbf{v})(\mathbf{x}-\mathbf{v})\;, \ 
\mathbf{y}=-\mathbf{x}+k(\mathbf{x},-\mathbf{v})(\mathbf{x}-\mathbf{v})\;,
\end{equation*}
and  to 
\begin{equation*}
\mathbf{v}=-\mathbf{u}+k(\mathbf{u},-\mathbf{y})(\mathbf{u}-\mathbf{y})\;, \ 
\mathbf{x}=-\mathbf{y}+k(\mathbf{u},-\mathbf{y})(\mathbf{u}-\mathbf{y})\;.
\end{equation*}
\end{proof}

  {
\begin{remark}
 If we replace the quadratic form $\langle\mathbf{x},\mathbf{y}\rangle=x_0 y_0-x_1 y_1$ that appears in the momentum-energy map \eqref{MEmap} with the scalar product 
$\langle\mathbf{x},\mathbf{y}\rangle=x_0 y_0+x_1 y_1$ of the Euclidean space $\mathbb{R}^2$, then we retrieve a map $R_k$ which (up to a permutation) was presented in \cite{Adler} as integrable deformations of a polygon. 
The two maps are not equivalent in the real space $\mathbb{R}^2\times \mathbb{R}^2$. The map $R_k$ preserves the circles $x_0^2+x_1^2=c_1^2$ and  $y_0^2+y_1^2=c_2^2$, for constants $c_1,c_2$, while the momentum-energy map $\mathbf{R}$ the parabolas $x_0^2-x_1^2=m_1^2c^2$ and 
$y_0^2-y_1^2=m_2^2c^2$. However, the two maps are equivalent by considering a complex change of variables. That is 
\begin{equation*} 
{R}_{k}= (\phi \times \phi) \circ {\mathbf{R}} \circ (\phi^{-1} \times \phi^{-1}) \;, 
\end{equation*} 
for $\phi(x_0,x_1)=(x_0,i x_1)$. 

\end{remark}
}

The momentum-energy Yang--Baxter map \eqref{MEmap} is reversible, which means that  $\mathbf{R}_{21}\circ \mathbf{R}=Id$ and  
it is an involution, 
$\mathbf{R}\circ \mathbf{R}=Id$. Moreover, it is associated with a system of 3D consistent lattice equation. We can derive this equation from proposition \ref{YBto3D} by considering the abelian group $(\mathbb{R}^2,+)$ and the 
invariant condition \eqref{invcond1}. In this case, equation \eqref{Shibf} becomes 
\begin{equation} \label{system3D}
\mathbf{w}_{n+1,m}-\mathbf{w}_{n,m+1}=k(\mathbf{w}_{n+1,m}-\mathbf{w}_{n,m},\mathbf{w}_{n+1,m+1}-\mathbf{w}_{n+1,m})(\mathbf{w}_{n+1,m+1}-\mathbf{w}_{n,m})\;. 
\end{equation}
We can derive this system directly from the  map \eqref{MEmap} by setting 
\begin{align*}
&\mathbf{x}=\mathbf{w}_{n+1,m}-\mathbf{w}_{n,m}, \ 
\ \ \mathbf{y}=\mathbf{w}_{n+1,m+1}-\mathbf{w}_{n+1,m}, \\ 
&\mathbf{u}=\mathbf{w}_{n+1,m+1}-\mathbf{w}_{n,m+1},  \ 
\mathbf{v}=\mathbf{w}_{n,m+1}-\mathbf{w}_{n,m}.
\end{align*}
For the vertex variables $\mathbf{w}_{n,m}=(f_{n,m},g_{n,m})$, the 
system \eqref{system3D} is expressed as
\begin{align} \label{sys21}
\frac{(f_{n+1,m+1}-f_{n,m+1})^2-(f_{n+1,m}-f_{n,m})^2}{(g_{n+1,m+1}-g_{n,m+1})^2-(g_{n+1,m}-g_{n,m})^2} &=1\;, \\
\frac{(f_{n+1,m+1}-f_{n+1,m})^2-(f_{n,m+1}-f_{n,m})^2} {(g_{n+1,m+1}-g_{n+1,m})^2-(g_{n,m+1}-g_{n,m})^2} &=1\;. \label{sys22}
\end{align}
  {
The vector equation \eqref{system3D}, or equivalently system \eqref{sys21}-\eqref{sys22}, is uniquely solvable with respect to any one of the arguments $\mathbf{w}_{n,m}=(f_{n,m},g_{n,m})$, $\mathbf{w}_{n+1,m}=(f_{n+1,m},g_{n+1,m})$, 
$\mathbf{w}_{n,m+1}=(f_{n,m+1},g_{n,m+1})$ and $\mathbf{w}_{n+1,m+1}=(f_{n+1,m+1},g_{n+1,m+1})$. Its 3D consistency follows from proposition \ref{YBto3D}. We can also obtain a Lax representation for this system from the Lax representation \eqref{LaxME} of the momentum-energy map, i.e. the solutions of \eqref{system3D} satisfy the equation 
$$\mathbf{L}(\mathbf{w}_{n+1,m+1}-\mathbf{w}_{n,m+1},\zeta)\mathbf{L}(\mathbf{w}_{n,m+1}-\mathbf{w}_{n,m},\zeta)=\mathbf{L}(\mathbf{w}_{n+1,m+1}-\mathbf{w}_{n+1,m},\zeta)\mathbf{L}(\mathbf{w}_{n+1,m}-\mathbf{w}_{n,m},\zeta).$$}


\section{Higher dimensional generalisations} \label{SecHigherGen}

The aim of this section is to generalise the results of the head-on collisions in higher dimensions. We will mainly focus on the two-dimensional space case. So, we consider that the two colliding particles, with rest masses $m_1$ and $m_2$, have two-dimensional vector momentums 
$\mathbf{p}_x=(p_{x_1},p_{x_2})$, $\mathbf{p}_y=(p_{y_1},p_{y_2})$ and 
$\mathbf{p}_u=(p_{u_1},p_{u_2})$, $\mathbf{p}_v=(p_{v_1},p_{v_2})$ before and after collision respectively. We also denote by 
$E_x$, $E_y$, and $E_u$, $E_v$ the corresponding relativistic energies of the particles and we define the momentum-energy  vectors  
\begin{align*}
&\mathbf{x}=(x_0,x_1,x_2):= 
(\frac{E_x}{c},\mathbf{p}_x)\;, \   \mathbf{y}=(y_0,y_1,y_2):= 
(\frac{E_y}{c},\mathbf{p}_y)\;,\\ 
&\mathbf{u}=(u_0,u_1,u_2):= 
(\frac{E_u}{c},\mathbf{p}_u)\;, \   \mathbf{v}=(v_0,v_1,v_2):=
(\frac{E_v}{c},\mathbf{p}_v)\;.
\end{align*}
The conservation of relativistic energy and momentum along with the energy-momentum relation imply the system 
\begin{equation} \label{system2}
\mathbf{u}+\mathbf{v}=\mathbf{x}+\mathbf{y}\;, \ 
\langle\mathbf{u},\mathbf{u}\rangle=\langle\mathbf{x},\mathbf{x}\rangle\;,\ \langle\mathbf{v},\mathbf{v}\rangle=\langle\mathbf{y},\mathbf{y}\rangle\;,
\end{equation}
where $\langle \ , \ \rangle$ denotes the quadratic form  $\langle\mathbf{x},\mathbf{y}\rangle:=x_0 y_0-x_1 y_1-x_2y_2$.

The system \eqref{system2} consists of five equations, so more information is needed to derive a unique solution with respect to $\mathbf{u}$ and $\mathbf{v}$. This could be, for instance, one of $u_i$, $v_i$, or the scattering angle of the two 
colliding particles after the collision, and can be linked to an extra parameter. 

\subsection{A 2D collision Yang--Baxter map }  

In this section, we present a higher dimensional Yang--Baxter map associated with relativistic collisions on the plane,  which satisfies an extra parametric invariant condition. This map  generalises the head-on collision Yang--Baxter maps of the previous section and corresponds to a system of 3D consistent lattice equations. 

\begin{theorem} \label{col2D}
The map $\mathbf{\hat{R}}_{\alpha,\beta}:(\mathbf{x},\mathbf{y})\mapsto (\mathbf{u},\mathbf{v})$, with  
\begin{align}
&\mathbf{u}= \ \hat k(\mathbf{x},\mathbf{y},q)\mathbf{x}+(1-\hat k(\mathbf{y},\mathbf{x},q))\mathbf{y}+2q \hat{\mathbf{r}}(\mathbf{x},\mathbf{y},q)\;, \label{u2d}\\ 
&\mathbf{v}=(1-\hat k(\mathbf{x},\mathbf{y},q))\mathbf{x}+\hat k(\mathbf{y},\mathbf{x},q)\mathbf{y}-2q \hat{\mathbf{r}}(\mathbf{x},\mathbf{y},q)  \label{v2d}\;,
\end{align}
where  
\begin{align*} \label{Kd2}
\hat k(\mathbf{x},\mathbf{y},q) =\frac{\langle\mathbf{x},\mathbf{x}\rangle-\langle\mathbf{y},\mathbf{y}\rangle+q^2}{\langle\mathbf{x}+\mathbf{y},\mathbf{x}+\mathbf{y}\rangle+q^2}\;, \ 
\hat{\mathbf{r}}(\mathbf{x},\mathbf{y},q) &=\frac{(x_1 y_2-x_2 y_1,x_0 y_2-x_2y_0,x_1 y_0-x_0 y_1)}{\langle\mathbf{x}+\mathbf{y},\mathbf{x}+\mathbf{y}\rangle+q^2}\;,
\end{align*}
 and $q=\alpha-\beta$, 
 is a parametric quadrirational Yang--Baxter map which satisfies the invariant conditions \eqref{system2} and admits a Lax representation with Lax matrix
 \begin{equation} \label{Lax2D}
  {\mathbf{L}(\mathbf{x},\alpha,\zeta)}=\left(
\begin{array}{cc}
  \zeta+i(x_2 -\alpha)  & \ x_0+x_1 \\
 x_0-x_1 &  \zeta-i(x_2+a)
\end{array}
\right).
\end{equation}
\end{theorem}

\

\begin{proof}
For every generic $2\times 2$ matrices $X$, $Y$, $U$ and $V$, the system (factorization system)  
$$(U-\zeta I) (V-\zeta I)=(Y-\zeta I)(X-\zeta I),$$ together with the equation 
$\det(U-\zeta I)=\det(X-\zeta I)$, where $I$ indicates the identity matrix, admits a non-trivial solution with respect to $U$ and $V$, that is 
$$U=(YX-\det(X) I)(Y+X-\mathrm{Tr}(X)\, I)^{-1}\;, \ V=Y+X-U\;. $$
The map  $\mathcal{R}_I:(X,Y)\mapsto (U,V)$ defined by this solution is a quadrirational Yang--Baxter map with Lax matrix $X-\zeta I$, and invariant conditions $\det(U)=\det(X)$, $\det(V)=\det(Y)$,  $\mathrm{Tr}(U)=\mathrm{Tr}(X)$ and $\mathrm{Tr}(V)=\mathrm{Tr}(Y)$. This fact is proved in \cite{kp1,kp3}.  

Now, if we express the generic elements of the matrix $X$ as 
 $$[X]_{11}=i(x_3+x_2), \ [X]_{12}=x_0+x_1, \ [X]_{21}=x_0-x_1, \ [X]_{22}=i(x_3-x_2)$$ and in a similar way the elements of the matrices $Y$, $U$ and $V$, 
then the map $\mathbf{\hat{R}}_{\alpha,\beta}$ constitutes the reduction of the map $\mathcal{R}_I$ to the invariant level sets 
\begin{equation} \label{symleavTr}
\{(x_1,x_2,x_3,x_4) \, | \, \mathrm{Tr}(X)=-2i \alpha\},  \ \{(y_1,y_2,y_3,y_4) \, | \, \mathrm{Tr}(Y)=-2i \beta\}, 
\end{equation}
by setting $x_3=u_3=-\alpha$ and $y_3=v_3=-\beta$. The same reduction at the Lax representation of $\mathcal{R}_I$ implies the Lax representation of $\mathbf{\hat{R}}_{\alpha,\beta}$ with Lax matrix \eqref{Lax2D}. Hence, the Yang--Baxter property, the Lax representation, the quadrirationality and the invariant conditions of $\mathbf{\hat{R}}_{\alpha,\beta}$  follow from the corresponding properties of the Yang--Baxter map $\mathcal{R}_I$ under this reduction.
\end{proof}

  { The parametric map $\mathbf{\hat{R}}_{\alpha,\beta}$ of theorem \ref{col2D} preserves the relativistic energy and momenta of the system of two colliding particles, as well as the rest masses $m_1$ and $m_2$. It is derived as the unique solution with respect to $\mathbf{u}$ and $\mathbf{v}$ of the factorization problem 
\begin{equation} \label{fac2d}
\mathbf{L}(\mathbf{u},\alpha,\zeta)\mathbf{L}(\mathbf{v},\beta,\zeta)=\mathbf{L}(\mathbf{x},\alpha,\zeta)\mathbf{L}(\mathbf{y},\beta,\zeta).
\end{equation}
In addition, from the Lax representation we can trace the extra parametric invariant condition 
\begin{equation*}
\alpha u_0+\beta v_0+(v_1 u_2-u_1 v_2)=\alpha x_0+\beta y_0+(x_1y_2-y_1x_2)\;,
\end{equation*}
or equivalently
\begin{equation} \label{R2inv}
\frac{\det(\mathbf{p}_u,\mathbf{p}_v)+\det(\mathbf{p}_x,\mathbf{p}_y)}{u_0-x_0}=q\;,
\end{equation}
for $q=\alpha-\beta$. 
On the other hand, we can derive directly \eqref{u2d}-\eqref{v2d} from the unique solution of the system of \eqref{system2} and 
\eqref{R2inv}, i.e.  \eqref{fac2d} is equivalent to \eqref{system2},\eqref{R2inv}. 
Hence, $\mathbf{\hat{R}}_{\alpha,\beta}$ represents two-dimensional relativistic collisions which satisfy the 
parametric invariant condition \eqref{R2inv}. }

The actual values of $\mathbf{\hat{R}}_{\alpha,\beta}$ depend only on the difference $q$ of the Yang--Baxter parameters $\alpha$ and $\beta$. The parameter $q$ can be related to the extra information which is required to fully determine planar collisions. For example, the scattering angle $\theta_L$  of the first particle with respect to the first axis is expressed  with respect to the parameter $q$ and the initial conditions $\mathbf{x}$, $\mathbf{y}$ by the equation 
\begin{align} \label{tan1}
\tan{\theta_L} =\frac{u_2} 
{u_1} 
=\frac{(|\mathbf{x}|^2-|\mathbf{y}|^2+q^2)x_2+2(|\mathbf{x}|^2+\langle\mathbf{x},\mathbf{y}\rangle)y_2+2q(x_1 y_0-x_0 y_1)}{(|\mathbf{x}|^2-|\mathbf{y}|^2+q^2)x_1+2(|\mathbf{x}|^2+\langle\mathbf{x},\mathbf{y}\rangle)y_1+2q(x_0 y_2-x_2 y_0)}\;,
\end{align}
where here $|\mathbf{x}|^2=\langle\mathbf{x},\mathbf{x}\rangle$ and $|\mathbf{y}|^2=\langle\mathbf{y},\mathbf{y}\rangle.$ 

We will further investigate the role of the parameters of the map $\mathbf{\hat{R}}_{\alpha,\beta}$ in a special case. So, we consider 
a frame (lab frame) in which the second particle is at rest and the velocity of the first particle is directed along the first axis, that is $y_1=y_2=x_2=0$ and $x_1>0$. In this case, 
taking into account the energy-momentum relation for the two particles, $|\mathbf{x}|^2=m_1^2 c^2$ and $|\mathbf{y}|^2=m_2^2 c^2$, \eqref{tan1} becomes 
\begin{align} \label{tan2}
\tan{\theta_L} =
\frac{2q y_0}{|\mathbf{x}|^2-|\mathbf{y}|^2+q^2}=\frac{2q m_2 c}{m_1^2 c^2-m_2^2 c^2+q^2}\;.
\end{align}
In addition, we can express the scattering angle $\theta_L$ in the lab frame with respect to the scattering angle $\theta$ in the center-of-momentum frame, where the total momentum of the system is zero, by the formula 
\begin{align} \label{tanfr}
\tan{\theta_L}=\frac{m_2 c^2 \sin{\theta}}{E_x'+E_y' \cos{\theta}}=
\frac{2 c^2 m_2 \tan{\frac{\theta}{2}}}{E_x'+E_y'+(E_x'-E_y')\tan^2{\frac{\theta}{2}}}
\;,
\end{align}
where here $E'_x$ and $E'_y$ denote the relativistic energies of the two particles before collision with respect to the center-of-momentum frame (so, $E_x'^2-E_y'^2=c^4(m_1^2-m_2^2)$\,). Comparing \eqref{tan2} and \eqref{tanfr}, we derive that 
\begin{equation} \label{qpar}
q=\frac{1}{c}(E'_x-E'_y)\tan{\frac{\theta}{2}}\;.
\end{equation}
The particles' after collision energies $E'_u$ and $E'_v$ with respect to the center of momentum frame are equal to $E'_x$ and $E'_y$ respectively. Thus, according to \eqref{qpar}, a natural choice for the Yang--Baxter parameters $\alpha,\beta$ in this case is 
\begin{equation} \label{abpar}
\alpha=\frac{E'_x}{c}\tan{\frac{\theta}{2}}=\frac{E'_u}{c}\tan{\frac{\theta}{2}}\;, \ 
\beta=\frac{E'_y}{c}\tan{\frac{\theta}{2}}=\frac{E'_v}{c}\tan{\frac{\theta}{2}}\;.
\end{equation}
Alternatively, if we denote by $W$ the invariant mass of the system of the particles defined by 
\begin{equation*}
 W^2c^2=\langle\mathbf{x}+\mathbf{y},\mathbf{x}+\mathbf{y}\rangle=\langle\mathbf{u}+\mathbf{v},\mathbf{u}+\mathbf{v}\rangle=
\frac{(E'_x+E'_y)^2}{c^2}\;,
\end{equation*}
then from \eqref{qpar} we can express the parameter $q$ in terms of $\theta$, $W$, $m_1$ and $m_2$, as  
\begin{equation*}
q=\frac{c}{W} (m_1^2-m_2^2)\tan{\frac{\theta}{2}}\;,    
\end{equation*}
and consider 
$$\alpha= \frac{cm_1^2}{W}\tan{\frac{\theta}{2}}\;, \ \beta= \frac{cm_2^2}{W}\tan{\frac{\theta}{2}}\;.$$

\subsection{Poisson structure, reductions and transfer dynamics } 

The construction of the Lax matrix \eqref{Lax2D}, as appears in the proof of theorem \ref{col2D}, suggests that 
it admits a compatible $r$-matrix Poisson structure, that is the Sklyanin bracket 
\cite{Skl1,Skl2}.
Indeed, the equation 
\begin{equation} \label{sklbr}
\{\mathbf{L}(\mathbf{x},  {\alpha},i \zeta ) \ \overset{\otimes }{,} \
\mathbf{L}(\mathbf{x},  {\alpha},i h)\}=[\frac{  {{P}}}{\zeta
-\eta},\mathbf{L}(\mathbf{x},  {\alpha},i \zeta )\otimes
\mathbf{L}(\mathbf{x},  {\alpha},i h)]\;,
\end{equation}
where $P$ denotes the permutation
operator,  $  {{P}}(x\otimes y) = y\otimes x$ and 
$\mathbf{L}$ the Lax matrix \eqref{Lax2D}, 
is equivalent to 
\begin{equation} \label{sklbr2}
\{x_0,x_1\}=x_2\;, \ \{x_0,x_2\}=-x_1\;, \  \{x_1,x_2\}=-x_0\;.
\end{equation}
  {
We can extend the Sklyanin bracket on $\mathcal{L} \times \mathcal{L}$, for $\mathcal{L}=\{\mathbf{L}(\mathbf{x},\alpha,\zeta )| \mathbf{x}\in \mathbb{R}^3,\alpha \in \mathbb{R} \}$
 by considering 
\begin{align*}
\{\mathbf{L}(\mathbf{x},{\alpha},i \zeta ) \ \overset{\otimes }{,} \
\mathbf{L}(\mathbf{x},{\alpha},i h)\}=[\frac{P}{\zeta
-\eta},\mathbf{L}(\mathbf{x},{\alpha},i \zeta )\otimes
\mathbf{L}(\mathbf{x},{\alpha},i h)]\;, \\
\{\mathbf{L}(\mathbf{y},{\beta},i \zeta ) \ \overset{\otimes }{,} \
\mathbf{L}(\mathbf{y},{\beta},i h)\}=[\frac{P}{\zeta
-\eta},\mathbf{L}(\mathbf{y},{\beta},i \zeta )\otimes
\mathbf{L}(\mathbf{y},{\beta},i h)]\;,
\end{align*}
and $\{\mathbf{L}(\mathbf{x},{\alpha},i \zeta ) \ \overset{\otimes }{,} \
\mathbf{L}(\mathbf{y},{\beta},i h)\}=0$, which corresponds to the Poisson tensor  
\begin{align*} \label{extPois1}
\pi_{(\mathbf{x},\mathbf{y})}&={x_{2}} \frac{\partial} {\partial x_{0}}\wedge \frac{\partial} {\partial x_{1}} -{x_{1}} \frac{\partial} {\partial x_{0}}\wedge \frac{\partial} {\partial x_{2}}-{x_{0}} \frac{\partial} {\partial x_{1}}\wedge \frac{\partial} {\partial x_{2}} \\
&+{y_{2}} \frac{\partial} {\partial y_{0}}\wedge \frac{\partial} {\partial y_{1}}-{y_{1}} \frac{\partial} {\partial y_{0}}\wedge \frac{\partial} {\partial y_{2}} -{y_{0}} \frac{\partial} {\partial y_{1}}\wedge \frac{\partial} {\partial y_{2}}\;
\end{align*}
on $\mathbb{R}^3 \times \mathbb{R}^3$. Now, we can show directly that the Yang--Baxter map $\mathbf{\hat{R}}_{\alpha,\beta}$ of theorem \ref{col2D} is Poisson with respect to $\pi_{(x,y)}$.}
This follows from the fact that $\mathbf{\hat{R}}_{\alpha,\beta}$\footnote{The Lax matrix of $\mathbf{\hat{R}}_{\alpha,\beta}$ belongs to the case I of the classification in \cite{kp3}} is a reduction on the Poisson submanifolds defined by the invariant level sets \eqref{symleavTr}, of the more general Poisson Yang--Baxter map $\mathcal{R}_I$ (see proof of theorem \ref{col2D}) with respect to the Sklyanin bracket \cite{kp1,kp3}. 

The Poisson structure \eqref{sklbr2} is of rank two and it admits the Casimir function
$$\mathcal{C}({\mathbf{x}})=\langle\mathbf{x},\mathbf{x}\rangle=x_0^2-x_1^2-x_2^2\;.$$
Correspondingly, the extended (rank-four) Poisson structure $\pi_{(\mathbf{x},\mathbf{y})}$ admits the Casimirs  
$\mathcal{C}({\mathbf{x}})$ and $\mathcal{C}({\mathbf{y}})$. Since the Yang--Baxter map $\mathbf{\hat{R}}_{\alpha,\beta}$ preserves the Casimirs, we can further reduce it to a symplectic Yang--Baxter map (with the masses as extra Yang--Baxter parameters) on the four-dimensional invariant symplectic leaves 
of $\pi_{(\mathbf{x},\mathbf{y})}$ defined by the connected components of 
$$\mathcal{S}_{m_1,m_2}=\{(\mathbf{x},\mathbf{y}) \, | \, \mathcal{C}({\mathbf{x}})=m_1^2c^2, \ \mathcal{C}({\mathbf{y}})=m_2^2c^2 \}\;.$$
However, the resulting map under this reduction is not rational.   

The head-on collision Yang--Baxter map \eqref{MEmap} is derived by reduction of $\mathbf{\hat{R}}_{\alpha,\beta}$ for $q=0$, i.e. $\alpha=\beta$, on 
an invariant manifold. Particularly, from \eqref{u2d}-\eqref{v2d} we observe  that if we set $x_2=y_2=0$, for $q=0$,  then we obtain $u_2=v_2=0$, which shows that $$\mathcal{M}= \{ (x_0,x_1,0,y_0,y_1,0) : x_i,y_i \in \mathbb{R} \}$$ is an invariant manifold of $\mathbf{\hat{R}}_{\alpha,\beta}$. The reduced map on $\mathcal{M}$  coincides with the head-on momentum-energy Yang--Baxter map \eqref{MEmap}.

The dynamical behaviour of a plain Yang--Baxter map is usually rather trivial. For example all the maps of the classification 
in \cite{ABS2,papclas} are involutions and the same is true for the maps \eqref{colYB}, \eqref{YBab} and \eqref{MEmap} (but not for the higher-dimensional Yang--Baxter map of theorem \ref{col2D}). Nevertheless, for any Yang--Baxter map various families of multidimensional maps, usually referred as {\it{transfer maps}}, can be generated that exhibit highly non-trivial behaviour. In \cite{ves2,ves3}, Veselov introduced an hierarchy of commuting transfer maps which preserve the spectrum of their monodromy matrix. The transfer maps of the collision Yang--Baxter maps represent particular sequences of $n$ colliding particles and their commutativity reflects the fact that the resulting  momentum-energy vectors are independent of the ordering of the collisions.

Here, we will focus on a variant of Veselov's transfer maps associated with periodic staircase initial value problems of integrable lattice equations \cite{PNC,QCPN} (see also \cite{kp2,Koul} for the case of Yang--Baxter maps). In this framework, we define the  transfer map of a Yang--Baxter map $R_{\alpha,\beta}$, as the  map  $$T_n:(x_1,x_2, \dots, x_n,y_1,y_2 \dots, y_n) \mapsto (x_1',x_2' \dots, x_n',y_2', y_3'\dots, y_n',y_1'),$$
where $(x_i',y_i')=R_{\alpha,\beta}(x_i,y_i)$, and the {\it $k$-transfer map} as the map 
$T_n^k:=\underbrace{T_n\circ \dots \circ T_n}_k$. We also define the {\it{monodromy matrix}} $M_n$ of $T_n$, with   
$$M_n(x_1,\dots,x_n,y_1, \dots, y_n)= \prod_{i=0}^{n-1}L(y_{n-i},\beta,\zeta)L(x_{n-i},\alpha,\zeta)\;.$$ 
From the definition of the monodromy matrix and the Lax representation \eqref{laxmat}, it follows that 
\begin{equation} \label{monfac}
L(y_1',\beta,\zeta) M_n(x_1,\dots,x_n,y_1, \dots, y_n)=M_n(T_n(x_1,\dots,x_n,y_1, \dots, y_n)) L(y_1',\beta,\zeta)\;.
\end{equation}
Hence, we can derive integrals of any transfer map from the spectrum of the 
monodromy matrix.
Similarly, we can show that in the more general case of (non-autonomous) transfer maps including different parameters $\alpha_i, \beta_i$,  where 
$(x_i',y_j'):=R_{\alpha_i,\beta_j}(x_i,y_j)$, the $n$-transfer map $T_n^n$ preserves the spectrum of the corresponding monodromy matrix.  

Let us denote by $\mathbf{T}_n$ the $6n$-dimensional transfer map of the Yang--Baxter map $\mathbf{\hat{R}}_{\alpha,\beta}$, and by $\mathbf{M}(\zeta)=\mathbf{M}_n(\mathbf{x}_1,\dots,\mathbf{x}_n,\mathbf{y}_1, \dots, \mathbf{y}_n)$ and 
$\mathbf{M'}(\zeta)=\mathbf{M}_n(\mathbf{T}_n(\mathbf{x}_1,\dots,\mathbf{x}_n,\mathbf{y}_1, \dots, \mathbf{y}_n))$ the corresponding monodromy matrices. 
By construction, we obtain three linear integrals of $\mathbf{T}_n$ (associated with the invariant condition $\mathbf{u}+\mathbf{v}=\mathbf{x}+\mathbf{y}$ of $\mathbf{\hat{R}}_{\alpha,\beta}$),
$$I_k=\sum_{i=1}^n(x_i+y_i)\;, \ k=0,1,2\;,$$
which represent the conservation of energy and momentum. More integrals are obtained from the spectrum of $\mathbf{M}(\zeta)$.

The comultiplication property of the Sklyanin bracket (see e.g. \cite{sklyBack,Tsigan} for the Sklyanin bracket with regard  to the Heisenberg magnetic chain) implies that 

\begin{equation*} 
\{\mathbf{M}(i \zeta ) \ \overset{\otimes }{,} \
\mathbf{M}(i h)\}=[\frac{{{P}}}{\zeta
-\eta},\mathbf{M}(i \zeta )\otimes
\mathbf{M}(i h)]\;,
\end{equation*}
and from \eqref{monfac} we derive that 
$\{\mathbf{M}'(i \zeta ) \ \overset{\otimes }{,} \
\mathbf{M}'(i h)\}=[\frac{  {{P}}}{\zeta
-\eta},\mathbf{M}'(i \zeta )\otimes
\mathbf{M}'(i h)]\;,$
which shows that $\mathbf{T}_n$ is a Poisson map with respect to 
the (extended) Sklyanin bracket on $\mathbb{R}^{6n}$,  
$$\pi_{(\mathbf{x}_1,\mathbf{y}_1)}+\pi_{(\mathbf{x}_2,\mathbf{y}_2)}+\dots +\pi_{(\mathbf{x}_n,\mathbf{y}_n)}\;.$$
Finally, the Sklyanin bracket ensures that the integrals obtained from the spectrum of  $\mathbf{M}(\zeta)$ are in involution
 \cite{BabVi}. 
We summarise all these results in the following proposition.
\begin{proposition}
The transfer map $\mathbf{T}_n$ is Poisson with respect to the Sklyanin bracket and preserves the spectrum of the monodromy matrix  $\mathbf{M}(\zeta)$, the corresponding $2n$ Casimirs $\mathcal{C}({\mathbf{x}_i})$, $\mathcal{C}({\mathbf{y}_i})$ and the three linear integrals $I_k$. Furthermore, $\{tr\mathbf{M}(\zeta),tr\mathbf{M}(\eta)\}=0$.    
\end{proposition}

Similar results hold also for Veselov's transfer maps of $\mathbf{\hat{R}}_{\alpha,\beta}$. In order to complete a proof of the Liouville integrability of the transfer maps we need to show that the spectrum of the monodromy matrix generates enough functional independent integrals. In the future we aim to investigate in detail the Liouville integrability of several types of collision transfer maps.

\subsection{The Lattice equation associated with 2D collisions }  

As in the case of head-on collisions, we can derive a $3$D-consistent system of lattice equations associated with the Yang--Baxter map $\mathbf{\hat{R}}_{\alpha,\beta}$ of theorem \ref{col2D}. Here, we  apply proposition \ref{YBto3D} on the abelian group $(\mathbb{R}^3,+)$  by considering the 
invariant condition $\mathbf{u}+\mathbf{v}=\mathbf{x}+\mathbf{y}$, for $\mathbf{u}$ and $\mathbf{v}$ given by \eqref{u2d} and \eqref{v2d} respectively.   
In this case, equation \eqref{Shibf} for the vector vertex variables $\mathbf{w}_{n,m}=(f_{n,m},g_{n,m},h_{n,m})$ implies the  
$3$D-consistent system 
\begin{align} \label{system23D}
\mathbf{w}_{n+1,m}-\mathbf{w}_{n,m+1}&=\hat{k}(\mathbf{w}_{n+1,m}-\mathbf{w}_{n,m},\mathbf{w}_{n+1,m+1}-\mathbf{w}_{n+1,m},q)(\mathbf{w}_{n+1,m}-\mathbf{w}_{n,m}) \nonumber \\ 
&-\hat{k}(\mathbf{w}_{n+1,m+1}-\mathbf{w}_{n+1,m},\mathbf{w}_{n+1,m}-\mathbf{w}_{n,m},q)
(\mathbf{w}_{n+1,m+1}-\mathbf{w}_{n+1,m}) \nonumber \\
&+ 2 q  \hat{\mathbf{r}}(\mathbf{w}_{n+1,m}-\mathbf{w}_{n,m},\mathbf{w}_{n+1,m+1}-\mathbf{w}_{n+1,m},q)  \;. 
\end{align}
Similarly to the head-on collision system \eqref{system3D}, we can derive system \eqref{system23D} directly from  \eqref{u2d}-\eqref{v2d} by setting 
$\mathbf{x}=\mathbf{w}_{n+1,m}-\mathbf{w}_{n,m}$,  
$\mathbf{y}=\mathbf{w}_{n+1,m+1}-\mathbf{w}_{n+1,m}$,  
$\mathbf{u}=\mathbf{w}_{n+1,m+1}-\mathbf{w}_{n,m+1}$ and  
$\mathbf{v}=\mathbf{w}_{n,m+1}-\mathbf{w}_{n,m}$. 

System \eqref{system23D} is uniquely solvable with respect to any one of the arguments $\mathbf{w}_{n,m}$, $\mathbf{w}_{n+1,m}$, 
$\mathbf{w}_{n,m+1}$, $\mathbf{w}_{n+1,m+1}$ and 3D consistent according to proposition 2.1.
Furthermore, it admits the Lax representation
\begin{align*}
&\mathbf{L}(\mathbf{w}_{n+1,m+1}-\mathbf{w}_{n,m+1},\alpha,\zeta)\mathbf{L}(\mathbf{w}_{n,m+1}-\mathbf{w}_{n,m},\beta,\zeta)\\
&=\mathbf{L}(\mathbf{w}_{n+1,m+1}-\mathbf{w}_{n+1,m},\beta,\zeta)\mathbf{L}(\mathbf{w}_{n+1,m}-\mathbf{w}_{n,m},\alpha,\zeta)\;.
\end{align*}
where $\mathbf{L}(\mathbf{w}_{n+1,m+1}-\mathbf{w}_{n,m+1},\alpha,\zeta)$ is the Lax matrix \eqref{Lax2D}. 
The system \eqref{system23D} is reduced to \eqref{system3D} for $q=0$ and $h_{n,m}=0$, for all $n,m \in \mathbb{Z}$. 

In some cases compatible Poisson structures of Yang--Baxter maps of particular form give rise to compatible Poisson structures for periodic reductions of the corresponding 3D-consistent lattice equations \cite{KT}. However, this connection is not completely clear yet and a straightforward implementation of the results in \cite{KT} do not apply here. 
We will defer this problem for a future work in which we intend to study  
invariant Poisson structures for the periodic reductions of relativistic collision quadrilateral equations.

\section{Conclusion} \label{SecConc}
In this paper, we presented discrete integrable systems, namely Yang--Baxter maps and 3D consistent systems of lattice equations, associated with 
elastic relativistic particle collisions. Our approach was based on the momentum-energy transformation which generalises the results of \cite{Koul} and implies quadrirational Yang--Baxter maps, with respect to the original energy and momentum variables, as well as affine linear quadrilateral equations. Parametric generalisations of these systems, regarding planar relativistic collisions, were also presented, along with an $r$-matrix formalism suitable to study the Liouville integrability of the transfer maps. A complete proof of the Liouville integrability, including the investigation of exact solutions of  transfer maps and of plane wave reductions of the induced lattice systems require further study. 

The transfer maps of the Yang--Baxter maps and the periodic reductions 
of the quadrilateral equations correspond to periodic boundary conditions on two-dimensional lattices. Nevertheless, a similar approach to \cite{CCZ,CZ} can be considered to study reflection maps and fixed boundary value problems for the discrete collision systems. Furthermore, various non-commutative and anti-commutative solutions of the Yang--Baxter equation appear in literature  \cite{AdamPap,ABS2,Doli,KasKoul,sokkoul}. It would be  interesting to examine the existence of non-commutative analogues of the collision Yang--Baxter maps and the corresponding 3D consistent equations. 

 { Proceeding to higher-dimensional collision problems requires extra information regarding the colliding particles  and equivalent parametric invariant conditions. The existence of higher-dimensional refactorization problems which reproduce such conditions, in addition to  the conservation of momentum-energy, and result in quadrirational maps and affine linear systems of  equations is left for future research. It is expected that higher-dimensional systems representing elastic collisions will be integrable as well.
However, the interpretation from a physics point of view of standard integrability features (e.g. soliton/breather solutions, symmetries and conservation laws) of discrete collision systems and their continuous counterparts, with regard to chains  of relativistic colliding particles, is not yet very clear to the author and deserves further investigation.}

\end{document}